\documentclass[letterpaper,10pt]{article}

\author{Daniel S. Roche\\
Computer Science Department\\
United States Naval Academy\\
Annapolis, Maryland, U.S.A.}

\def\tit{Error correction in fast matrix multiplication and inverse}
\def\auts{Daniel S.\ Roche}

\usepackage{multirow}
\usepackage[ruled,boxed,linesnumbered,vlined,noend]{algorithm2e}
\usepackage{etoolbox}

\newtoggle{sage}

\togglefalse{sage}

\usepackage[utf8]{inputenc}
\usepackage[T1]{fontenc}
\usepackage[english]{babel}
\usepackage[svgnames]{xcolor}
\usepackage{xstring}
\usepackage{amsmath}
\usepackage{amssymb}
\usepackage{amsthm}
\usepackage{mathabx}
\usepackage{mathtools}
\usepackage{paralist}
\usepackage{fancyhdr}
\usepackage{lastpage}
\usepackage[longnamesfirst]{natbib}
\usepackage[backref=page,pdfpagelabels=true]{hyperref}

\hypersetup{
  pdfauthor={\auts},
  pdftitle={\tit},
  colorlinks,
  linkcolor=DarkBlue,
  citecolor=DarkGreen,
  urlcolor=DarkBlue,
  bookmarksnumbered,
}
\renewcommand{\backref}[1]{Referenced on %
  \expandarg\StrCount{#1}{,}[\ncommas]%
  \ifthenelse{\ncommas = 0}{page #1}%
  {pages \StrBefore[\ncommas]{#1}{,}\ and\StrBehind[\ncommas]{#1}{,}}%
.}

\newcommand{\doi}[1]{doi: \href{http://dx.doi.org/#1}{\path{#1}}}

\numberwithin{equation}{section}

\title{\tit}

\usepackage{enumitem}
\usepackage[capitalise,noabbrev,nameinlink]{cleveref}
\Crefname{enumi}{Step}{Steps}
\usepackage{todonotes}

\newcommand{\nzrowstext}{\texttt{FindNonzeroRows}}
\newcommand{\nzrows}{\texorpdfstring{\hyperlink{nzrows}{\nzrowstext{}}}{\nzrowstext{}}}
\newcommand{\rootfindtext}{\texttt{FindRoots}}
\newcommand{\rootfind}{\texorpdfstring{\hyperlink{rootfind}{\rootfindtext{}}}{\rootfindtext{}}}
\newcommand{\spinterptext}{\texttt{MultiSparseInterp}}
\newcommand{\spinterp}{\texorpdfstring{\hyperlink{spinterp}{\spinterptext{}}}{\spinterptext{}}}
\newcommand{\diffevaltext}{\texttt{DiffEval}}
\newcommand{\diffeval}{\texorpdfstring{\hyperlink{diffeval}{\diffevaltext{}}}{\diffevaltext{}}}
\newcommand{\mulerrtext}{\texttt{MultiplyEC}}
\newcommand{\mulerr}{\texorpdfstring{\hyperlink{mulerr}{\mulerrtext{}}}{\mulerrtext{}}}
\newcommand{\inverserrtext}{\texttt{InverseEC}}
\newcommand{\inverserr}{\texorpdfstring{\hyperlink{inverserr}{\inverserrtext{}}}{\inverserrtext{}}}

\newcommand{\expons}{\ensuremath{\mathcal{S}}}

\newcommand{\rowinds}{\mathcal{J}}
\newcommand{\rootlist}{\mathcal{Q}}
\newcommand{\polyroots}{\mathcal{R}}

\newcommand{\ZZ}{\ensuremath{\mathbb{Z}}}
\newcommand{\NN}{\ensuremath{\mathbb{N}}}
\newcommand{\QQ}{\ensuremath{\mathbb{Q}}}
\newcommand{\RR}{\ensuremath{\mathbb{R}}}
\newcommand{\CC}{\ensuremath{\mathbb{C}}}
\newcommand{\F}{\ensuremath{\mathsf{F}}}

\newcommand{\GF}[1]{\ensuremath{\mathbb{F}_{#1}}}

\newcommand{\bnd}[2]{\ensuremath{#1\mathopen{}\left(#2\right)\mathclose{}}}
\newcommand{\oh}[1]{\bnd{O}{#1}}
\newcommand{\softoh}[1]{\bnd{\widetilde{O}}{#1}}
\newcommand{\ceil}[1]{\ensuremath{\left\lceil#1\right\rceil}}
\newcommand{\floor}[1]{\ensuremath{\left\lfloor#1\right\rfloor}}
\newcommand{\supp}{\ensuremath{\mathsf{supp}}}

\newcommand{\vv}{\ensuremath{\mathbf{v}}}
\newcommand{\xx}{\ensuremath{\mathbf{x}}}
\newcommand{\uu}{\ensuremath{\mathbf{u}}}

\DeclareMathOperator{\llog}{loglog}

\newtheorem{theorem}{Theorem}[section]
\newtheorem{lemma}[theorem]{Lemma}

\newtheorem{fact}[theorem]{Fact}
\newtheorem{cor}[theorem]{Corollary}

\DontPrintSemicolon
\SetKw{break}{break}

\begin{document}
\maketitle

\begin{abstract}
  We present new algorithms to detect and correct errors in the
  product of two matrices, or the inverse of a matrix, over an
  arbitrary field. Our algorithms do not require any additional
  information or encoding other than the original inputs and the erroneous
  output. Their running time is softly linear in the number of nonzero entries in
  these matrices when the number of errors is sufficiently small, and
  they also incorporate fast matrix multiplication so that the cost scales
  well when the number of errors is large.
  These algorithms build on the recent result of \citet{GLL+17} on
  correcting matrix products, as well as existing work on verification
  algorithms, sparse low-rank linear algebra, and sparse polynomial
  interpolation.
\end{abstract}

\section{Introduction}

Efficiently and gracefully handling computational errors is critical in
modern computing. Such errors can result from short-lived hardware
failures, communication noise, buggy software, or even malicious
third-party servers.

The first goal for fault-tolerant computing is \emph{verification}.
\citet{Fre79} presented a
linear-time algorithm to verify the correctness of a single
matrix product. Recently, efficient verification algorithms for
a wide range of computational linear algebra
problems have been developed
\citep{KNS11,DK14,DKTV16,DLP17}.

A higher level of fault tolerance is achieved through \emph{error
correction}, which is the goal of the present study. This is a strictly
stronger model than verification, where we seek not only to identify
when a result is incorrect, but also to compute the correct result if it
is ``close'' to the given, incorrect one.
Some recent error-correction problems considered in the computer algebra
literature include Chinese remaindering \citep{GRS99,KPR+10,BDFP15},
system solving \citep{BK14,KPSW17},
and function recovery \citep{CKP12,KY13}.

A generic approach to correcting errors \emph{in transmission}
would be to use an error-correcting code, writing the result with some
redundancy to support transmission over a noisy channel.
But this approach requires extra bandwidth, is limited to a fixed number
of errors, and does not account for malicious alterations or mistakes in
the computation itself.

Instead, we will use the information in the
problem itself to correct  errors. This is always possible by
re-computing the result, and so the goal is always to
correct a small number of errors in (much) less time than it would take
to do the entire computation again. Our notion of
error correction makes it
possible to correct an unbounded number of mistakes that occur either in
computation or transmission, either at random or by a malicious party.

Some other applications of error correction have nothing to do with
alterations or mistakes. One example is sparse
matrix multiplication, which can be solved with an error-correction
algorithm by simply assuming the ``erroneous'' matrix is zero. Another
example is computing a result over $\ZZ$ using Chinese remaindering
modulo small primes, where the entries have varying bit lengths. The
small entries are determined modulo the first few primes, and the
remaining large entries are found more quickly by applying an error
correction algorithm using (known) small entries.

\subsection{Our results}

\paragraph{Matrix multiplication with errors}
We consider two computational error correction problems. The first
problem is the same as in \citet{GLL+17}, correcting a matrix product.
Given $A,B,C\in\F^{n\times n}$, the task is to compute the unique
error matrix $E\in\F^{n\times n}$ such that $AB=C-E$.

\Cref{alg:mulerr} solves this problem using
$$\softoh{t + k\cdot \min\left(\ceil{\tfrac{t}{r}},
  n/\min(r,\tfrac{k}{r})^{3-\omega}\right)}$$
field operations, where
\begin{itemize}[nosep]
\item $t = \#A + \#B + \#C$ is the number of nonzero entries in
  the input matrices, at most $O(n^2)$;
\item $k=\#E$ is the number of errors in the given product $C$; and
\item $r\le n$ is the number of distinct rows (or columns) in
  which errors occur, whichever is larger.
\end{itemize}

\begin{table}[b]
  \centering\begin{tabular}{c|c|c|c}
  & \multicolumn{2}{c|}{Spread-out errors} & Compact errors \\
  & \multicolumn{2}{c|}{$r = \min(n,k)$} & $r=\sqrt{k}$ \\
  & $k\le n$ & $k > n$ & \\
  \hline &&& \\[-10pt]
  Sparse
    & $t+k+n$
    & $k\ceil{\tfrac{t}{n}} + n$
    & $\sqrt{k} t + k + n$ \\[2pt]
  \hline &&& \\[-10pt]
  \multirow{2}{*}{Dense}
    & \multirow{2}{*}{$t+kn$}
    & $k^{\omega-2}n^{4-\omega}$
    & $t + k^{(\omega-1)/2} n$ \\
    &
    & $k^{0.38}n^{1.63}$
    & $t + k^{0.69}n$
  \end{tabular}
\caption{Soft-oh cost of \cref{alg:mulerr} in different cases.
The two rows indicate whether sparse
or dense rectangular multiplication is used as a subroutine.}
\label{tab:mul}
\end{table}

An even more detailed complexity statement, incorporating rectangular
matrices and a failure probability, can be found in
\cref{thm:mulerr}.
To understand the implications of this complexity, consider
six cases, as summarized in \cref{tab:mul}.

The cost of our algorithm depends on the locations of the errors.
If there are a constant number of errors per row or column,
we can use sparse multiplication to achieve softly-linear
complexity $\softoh{t+n+k}$. This is similar to
\citep{GLL+17} except that we do not require the input matrices to be
dense.

If $k>n$ errors are spread out among all
rows and columns, the worst-case cost
using dense multiplication is $\softoh{k^{0.38}n^{1.63}}$. As the number
of errors grows to $n^2$, this cost approaches that of normal matrix
multiplication without error correction.

In the other extreme, if the errors are isolated to a square (but not
necessarily contiguous) submatrix, then the worst-case cost using dense
multiplication is $\softoh{t = k^{0.69}n}$. Again this scales to
$n^\omega$ as $k \to n^2$, but the cost is better for $k<n^2$.
This situation may also make sense in the context of distributed
computing, where each node in a computing cluster might be assigned to
compute one submatrix of the product.

Our algortihm is never (asymptotically) slower than that of
\citep{GLL+17} or the $O(n^\omega)$ cost of na\"ive recomputation, but
it is faster in many cases: for example when the inputs are sparse, when there are
many errors, and when the errors are compactly located.

\paragraph{Matrix inverse with errors}
The second problem we consider is that of correcting computational
errors in a matrix inverse. Formally, given $A,B\in\F^{n\times n}$,
the task is to compute the unique error matrix $E\in\F^{n\times n}$ such that
$A^{-1}=B+E$.

\Cref{alg:inverserr} shows how to solve this problem using
$$\softoh{t + nk/\min(r,\tfrac{k}{r})^{3-\omega} + r^\omega}$$
field operations, where the parameters $t, k, n$ are the same as before,
and $r$ is the number of rows or columns which contain errors, whichever
is \emph{smaller}.
This is the same complexity as our algorithm for
matrix multiplication with errors, plus an additional $r^\omega$ term.

The additional $r^\omega$ term makes the situation here less fractured than
with the multiplication algorithm.
In the case of spread-out errors, the cost is
$\softoh{t+nk+k^\omega}$, which becomes simply
$\softoh{n^2+k^\omega}$ when the input matrices are dense.
The case of compact errors is better, and the complexity is just
$\softoh{t+nk^{(\omega-1)/2}}$, the same as that of our multiplication
algorithm

As before, the cost of this algorithm
approaches the $O(n^\omega)$ cost of the na\"ive inverse computation
as the number of errors $k$ grows to $n^2$.

\paragraph{Domains}
The algorithms we present work over any field and require only that a
field element of order larger than $n$ can be found.
Because our algorithms explicitly take $O(n)$ time to compute the powers
of this element anyway, the usual difficulties with efficiently generating
high-order elements in finite fields do not arise.

Over the integers $\ZZ$ or rationals $\QQ$, the most efficient
approach would be to work modulo a series of primes and use Chinese
remaindering to recover the result. A similar approach could be used
over polynomial rings.

Over the reals $\RR$ or complex numbers $\CC$, our algorithms work only
under the unrealistic assumption of infinite precision.
We have not
considered the interesting question of how to recover from the two
types of errors (noise and outliers) that may occur in finite precision.

\subsection{Related work}

Let $\omega$ be a constant between 2 and 3 such that two $n\times n$
matrices can be multiplied using $O(n^\omega)$ field operations.
In practice, $\omega=3$ for
dimensions $n$ up to a few hundred, after which Strassen-Winograd
multiplication is used giving $\omega = < 2.81$.
The best asymptotic algorithm currently gives
$\omega < 2.3728639$ \citep{LGal14}.

Even though the practical
value of $\omega$ is larger than the asymptotic one, matrix
multiplication routines have been the subject of intense implementation
development for decades, and highly-tuned software is readily available
for a variety of architectures \citep{DGP08,WZZY13}.
Running times based on $n^\omega$ have
practical as well as theoretical significance; the $\omega$ indicates
where an algorithm is able to take advantage of fast low-level matrix
multiplication routines.

Special routines for multiplying sparse matrices have also been
developed.
Writing $t$ for the number of nonzero entries in the input, the standard row- or
column-wise sparse matrix multiplication costs $O(tn)$ field operations.
\citet{YZ05} improved this to
$\softoh{n^2 + t^{0.697}n^{1.197}}$.
Later work also incorporates the number $k$ of
nonzeros in the product matrix. \citet{Lin09} gave an output-sensitive
algorithm with running time $O(k^{0.186}n^2)$.
\citet{AP09} have another with complexity
$\softoh{t^{0.667}k^{0.667} + t^{0.862}k^{0.408}}$, an improvement when
the outputs are not too dense.
\citet{Pag13} developed a different approach with complexity
$\softoh{t+kn}$; that paper also includes a nice summary of the state of
the art. Our multiplication with errors algorithm can also be used for
sparse multiplication, and it provides a small improvement when the
input and output are both not \emph{too} sparse.

The main predecessor to our work is the recent algorithm of
\citet{GLL+17}, which can correct up
to $k$ errors in the product of two $n\times n$ matrices
using $\softoh{n^2 + kn}$ field operations. Their approach is based
in part on the earlier work of \citet{Pag13} and makes clever use of
hashing and fast Fourier transforms
in order to achieve the stated complexity.

Computing the inverse of a matrix has the same asymptotic complexity
$O(n^\omega)$ as fast matrix multiplication \citep{BH74}. In practice,
using recent versions of the FFPACK library \citep{DGP08} on an Intel desktop
computer, we find that a single $n\times n$ matrix inversion is only
about 33\% more expensive than a matrix-matrix product of the same size.

\section{Overview}

\paragraph{Multiplication with errors}\label{sec:over:mult}

A high-level summary of our multiplication algorithm is as follows:

\begin{enumerate}[noitemsep]
  \item \label{mulover:nzrows}
    Determine which rows in the product contain errors.
  \item
    Write the unknown sparse matrix of
    errors as $E = C - AB$. Remove the rows with no errors from $E, C,$
    and $A$,
    so we have $E' = C'-A'B$.
  \item \label{mulover:evals}
    Treating the rows of $E'$ as sparse polynomials, use structured
    linear algebra and fast matrix multiplication to
    evaluate each row polynomial at a small number of points.
  \item \label{mulover:spinterp}
    Use sparse polynomial interpolation to recover at least half of the
    rows from their evaluations.
  \item Update $C$ and iterate $O(\log k)$ times to recover all $k$
    errors.
\end{enumerate}

To determine the rows of $E$ which are nonzero in \cref{mulover:nzrows},
we apply a simple
variant of Frievalds' verification algorithm: for a random vector $\vv$,
the nonzero rows of $E\vv$ are probably the same as the nonzero rows of
$E$ itself.

The evaluation/interpolation approach that follows is reminiscent of that in
\citet{GLL+17}, with two important differences.
First, using sparse polynomial techniques rather than hashing and FFTs
means the complexity scales linearly with the number of nonzeros $t$ in the
input rather than the dense matrix size $n^2$. Second, by treating the
rows separately rather than recovering all errors at once, we are able
to incorporate fast matrix multiplication so that our worst-case cost
when all entries are erroneous is
never more than $O(n^\omega)$, the same as that of na\"ive
recomputation.

\Cref{sec:nzrows} provides details on the Monte Carlo algorithm to
determine the rows and columns where errors occur (\cref{mulover:nzrows}
above), \cref{sec:spinterp} presents a variant of sparse polynomial
interpolation that suits our needs for \cref{mulover:spinterp},
and \cref{sec:evals} explains how to perform the row
polynomial evaluations in \cref{mulover:evals}.
A certain rectangular matrix multiplication in this step turns out to
dominate the overall complexity.
We connect all these pieces of the
multiplication algorithm in
\cref{sec:mul}.
\iftoggle{sage}{Finally, \cref{sec:imp} shows some preliminary
results of a Sage implementation of our algorithm.}{}

\paragraph{Inverse with errors}\label{sec:over:inverse}

Our algorithm for correcting errors in a matrix inverse is based on the
multiplication algorithm, but with two complications. Writing $A$ for the
input matrix, $B$ for the erroneous inverse, and $E$ for the (sparse)
matrix of errors in $B$, we have:
$$AE = I - AB, \qquad EA = I - BA.$$

We want to compute $E\vv$ for a random vector $\vv$
in order to determine which rows of $E$ are nonzero. This is no longer
possible directly, but using the second formulation as above, we can
compute $EA\vv$. Because $A$ must be nonsingular, $A\vv$ has the same
distribution as a random vector, so this approach still works.

The second complication is that removing zero rows of $E$ does not change the
dimension of the right-hand side, but only removes corresponding
\emph{columns} from $A$ on the left-hand side.
We make use of the recent result of \citet{CKL13}, which shows how to
quickly find a maximal set of linearly independent rows in a sparse
matrix. This allows us to find a small submatrix $X$ of $A$ such that
$XE' = I' - A'B$, with $E', I', A'$ being compressed matrices as
before. Because the size of $X$ depends on the number of errors, it can
be inverted much more quickly than re-computing the entire inverse of
$A$.

The dominating cost in most cases is the same rectangular matrix product
as needed as in matrix multiplication with error correction. However,
here we also have the extra cost of finding $X$ and computing its
inverse, which dominates the complexity when there are a moderate number
of errors and their locations are spread out.

This algorithm --- which depends on the subroutines of
\cref{sec:nzrows,sec:spinterp,sec:evals} --- is presented in detail in
\cref{sec:inverse}.
\iftoggle{sage}{We also discuss
a preliminary Sage implementation of in \cref{sec:imp}}{}.

\section{Notation and preliminaries}

The soft-oh notation $\softoh{\cdots}$ is the same as the usual big-oh
notation but ignoring sub-logarithmic factors:
$f \in \softoh{g}$ if and only if $f\in \oh{g \log^{O(1)} g}$, for some
runtime functions $f$ and $g$.

We write $\NN$ for the set of nonnegative integers and
$\GF{q}$ for the finite field with $q$ elements.

For a finite set $\expons$, the number of elements in $\expons$
is denoted $\#\expons$, and
$\mathbb{P}(\expons)$ is the powerset of
$\expons$, i.e., the set of subsets of $\expons$.

The number of nonzero entries in a matrix $M\in\F^{m\times n}$ is written as $\#M$.
Note that $\#M \le mn$, and if $M$ has rank $k$ then $\#M \ge k$.

For a univariate polynomial $f\in\F[x]$,
$\supp(f) \subseteq \NN$ denotes the exponents of nonzero terms of $f$.
Note that $\#\supp(f) \le \deg f + 1$.

We assume that the number of field operations
to multiply two (dense) polynomials in $\F[x]$ with degrees less than
$n$ is $\softoh{n}$.
This is justified by the generic algorithm from \citet{CK91}
which gives $\oh{n\log n \llog n}$, or more result results of
\citet{HHL17} which improve this to $\oh{n\log n\, 8^{\log^* n}}$
for finite fields.

As discussed earlier, we take $\omega$ with $2\le \omega \le 3$ to be
any feasible exponent of matrix multiplication.
By applying fast matrix multiplication in blocks, it is also
possible to improve the speed of rectangular matrix multiplication in a
straightforward way. In particular:

\begin{fact}\label{thm:rectmul}
The product of any $m\times\ell$ matrix
times a $\ell\times n$ matrix can be found using
$\oh{m\ell n / \min(m,\ell,n)^{3-\omega}}$ ring operations.
\end{fact}

Note that it might be
possible to improve this for certain values of $m,\ell,n$ using
\cite{LGal12}, but that complexity is much harder to state and we have
not investigated such an approach.

We restate two elementary facts on sparse matrix multiplication.

\begin{fact}\label{thm:sparseapply}
  For any $M\in\F^{m\times n}$ and $\vv\in\F^n$, the matrix-vector
  product $M\vv$ can
  be computed using $O(\#M)$ field operations.
\end{fact}

\begin{cor}\label{thm:sparsemul}
  For any $A\in\F^{m\times \ell}$ and $B\in\F^{\ell \times n}$,
  their
  product $AB$ can be computed using $O(\#A\cdot n)$ field operations.
\end{cor}

\section{Identifying nonzero rows and columns}\label{sec:nzrows}

Our algorithms for error correction in matrix product and inverse
computation both begin by
determining which rows and columns contain errors.
This is accomplished in turn by applying a random row or column vector
to the unknown error matrix $E$.

For now, we treat the error matrix $E$ as a \emph{black box}, i.e., an
oracle matrix which we can multiply by a vector on the right or
left-hand side. The details of the black box construction differ for the
matrix multiple and inverse problem, so we delay those until later.

The idea is to apply a row or column vector of random values to the
unknown matrix $E$, and then examine which entires in the resulting
vector are zero. This is exactly the same idea as the classic Monte
Carlo verification algorithm of \citet{Fre79}, which we extend to which
rows or columns in a matrix are nonzero. This black-box procedure is
detailed in \cref{alg:nzrows} \nzrows{}.

\begin{algorithm}[htbp]
  \caption{\protect\hypertarget{nzrows}{\nzrowstext{}}%
    $(V\mapsto MV, \epsilon)$
  \label{alg:nzrows}}
  \KwIn{Black box for right-multiplication by an unknown matrix
    $M\in\F^{m\times n}$ and error bound $\epsilon\in\RR, 0<\epsilon<1$}
  \KwOut{Set $\rowinds\subseteq\{0,\ldots,m-1\}$ that with probability
    at least $1-\epsilon$ contains the index of all nonzero rows in
    $M$}
  \BlankLine
  $\ell \gets \ceil{\log_{\#\F}\tfrac{m}{\epsilon}}$ \;
  $V \gets$ matrix in $\F^{n\times \ell}$ with uniform random entries
    from $\F$ \;
  $B \gets MV$ using the black box \;
  \Return{Indices of rows of $B$ which are not all zeros}
\end{algorithm}

The running time of \nzrows{} depends entirely on the size of $\F$
and the cost of performing one black-box apply over the field $\F$. The correctness
depends on the parameter $\epsilon$.

\begin{lemma}\label{thm:nzrows}
  \nzrows{} always returns a
  list of nonzero row indices in $M$
  With probability at least $1-\epsilon$,
  it returns the index of \emph{every} nonzero row in $M$.
\end{lemma}
\begin{proof}
  Consider a single row $\uu^T$ of $M$.
  If $\uu$ is zero, then the corresponding row of $MV$ will always be
  zero; hence every index returned must be a nonzero row.

  Next assume $\uu^T$ has at least one nonzero entry, say at
  index $i$, and
  consider one row $\vv\in\F^n$ of the random matrix $V$.
  Whatever the other entries in $\uu$ and $\vv$ are, there is exactly one value
  for the $i$'th entry of $v$ such that $\uu^T\vv=0$.
  So the probability that $\uu^T\vv=0$ is $\tfrac{1}{\#\F}$,
  and by independence the probability that $\uu^TV = \mathbf{0}^{1\times
  \ell}$ is $1/(\#\F)^\ell$.

  Because $M$ has at most $m$ nonzero rows, the union bound tells us
  that the probability \emph{any one} of the corresponding rows of $MV$
  is zero is at most $m/(\#F)^\ell$. From the definition of $\ell$, this
  ensures a failure probability of at most $\epsilon$.
\end{proof}

Because we mostly have in mind the case that $\F$ is a finite field, we
assume in \nzrows{} and in the preceding proof
that it is possible to sample uniformly from the
entire field $\F$. However, the same results would hold if sampling from
any fixed subset of $\F$ (perhaps increasing the size of $\ell$ if the
subset is small).

\section{Batched low-degree sparse interpolation}\label{sec:spinterp}

Our algorithms use techniques from sparse polynomial
interpolation to find the
locations and values of erroneous entries in a matrix product or
inverse.

Consider an unknown univariate polynomial $f\in \F[x]$ with degree less
than $n$ and at most $s$ nonzero terms.
The nonzero terms of $f$ can be uniquely recovered from evaluations
$f(1),f(\theta),\ldots,f(\theta^{2s-1})$ at $2s$ consecutive powers
\citep{BT88,KY89}.

Doing this efficiently in our context
requires a few variations to the typical algorithmic approach.
While the state of the art seen in recent papers such as
\citep{JM10,Kal10a,AGR15,HG17} achieves softly-linear complexity in
terms of the sparsity $t$, there is either a higher dependency on the
degree, a restriction to certain fields, or both.

Here, we show how to take advantage of two unique aspects of the problem
in our context.
First, we will interpolate a \emph{batch} of sparse polynomials at
once, on the same evaluation points, which allows for some
useful precomputation.
Secondly, the maximum degree of any polynomial we interpolation is $n$, the
column dimension of the unknown matrix, and we are allowed linear-time
in $n$.

The closest situation in the literature is that of
\citet{HL13}, who show
how to recover a sparse polynomial if the exponents are known.
In our case, we can tolerate a much larger set containing the possible
exponents by effectively amortizing the size of this exponent set over
the group of batched sparse interpolations.

Recall the general outline of Prony's method
for sparse interpolation of $f\in\F[x]$ from
evaluations $(f(\theta^i))_{0\le i < 2s}$:
\begin{enumerate}[noitemsep]
  \item\label{step:prony:bm}
    Find the minimum polynomial $\Gamma\in\F[z]$ of the sequence of
    evaluations.
  \item \label{step:prony:roots}
    Find the roots $v_1,\ldots,v_s$ of $\Gamma$.
  \item \label{step:prony:dlog}
    Compute the discrete logarithm base $\theta$ of each $v_i$
    to determine the exponents of $f$.
  \item \label{step:prony:vand}
    Solve a transposed Vandermonde system built from $v_i$ and the
    evaluations to recover the coefficients of $f$.
\end{enumerate}

\Cref{step:prony:bm,step:prony:vand} can be performed over
any field using fast multipoint evaluation and interpolation in
$\softoh{s}$ field operations \citep{KY89}. However, the other steps
depend on the field and in particular \cref{step:prony:dlog} can be
problematic over large finite fields.

To avoid this issue, we first pre-compute all possible roots $\theta^{e_i}$
for any minpoly $\Gamma_i$ in any of the batch of $r$ sparse
interpolations that will be performed. Although the set of possible
roots is larger than the degree of any single minpoly $\Gamma_i$, in our
case the set is always bounded by the dimension of the matrix, so
performing this precomputation once is worthwhile.

\subsection{Batched root finding}

While fast root-finding procedures over many fields have already been
developed, we present a simple but efficient root-finding procedure for our
specific case that has the advantage of running in softly-linear time
over any coefficient field. The general idea is to first compute a
coprime basis (also called a gcd-free basis) for the minpolys $\Gamma_i$,
then perform multi-point evaluation with the precomputed list of
possible roots.

The details of procedure \rootfind{} are in \cref{alg:rootfind} below.
In the algorithm, we use a \emph{product tree} constructed from a list
of $t$ polynomials. This is a binary tree of height $\oh{\log t}$, with
the original polynomials at the leaves, and where every internal node is
the product of its two children. In particular, the root of the tree is
the product of all $t$ polynomials. The total size of a product tree,
and the cost of computing it, is softly-linear in the total size of the
input polynomials; see \citet{BM75} for more details.

\begin{algorithm}[htbp]
  \caption{\protect\hypertarget{rootfind}{\rootfindtext{}}%
    $(\rootlist, [\Gamma_1,\ldots,\Gamma_r])$%
    \label{alg:rootfind}}
  \KwIn{Set of possible roots $\rootlist\subseteq\F$ and $r$ polynomials
    $\Gamma_1,\ldots,\Gamma_r\in\F[z]$}
  \KwOut{A list of sets $[\polyroots_1,\ldots,\polyroots_r] \in \mathbb{P}(\rootlist)^r$
    such that
    $g_i(\alpha_j)=0$ for every $1\le i \le r$ and $\alpha_j \in \polyroots_i$}
  \BlankLine

  $\{g_0,\ldots,g_{t-1}\} \gets$
  coprime basis of $\{\Gamma_i\}_{1\le i \le r}$
  using Algorithm 18.1 of \citet{Ber05} \;

  \tcc{Compute product tree from bottom-up}
  \lFor{$j\gets 0,1,\ldots,t-1$}{$M_{0,j} \gets g_j$}
  \For{$i\gets 1,2,\ldots,\ceil{\log_2 t}$}{
    \For{$j\gets 0,1,\ldots,\ceil{t/2^i}-1$}{
      $M_{i,j} \gets M_{i-1,2j} \cdot M_{i-1,2j+1}$ \;
    }
  }

  \tcc{Find roots of basis elements}
  $S_{\ceil{\log_2 t}+1,0} \gets \rootlist$ \;
  \For{$i\gets \ceil{\log_2 t}, \ceil{\log_2 t}-1, \ldots, 0$}{
    \For{$j\gets 0,1,\ldots,\ceil{t/2^i}-1$}{
      Evaluate $M_{i,j}$ at each point in $S_{i+1,\floor{j/2}}$ using fast
        multipoint evaluation \label{rootfind:evals}\;
      $S_{i,j} \gets$ roots of $M_{i,j}$ found on previous step \;
    }
  }

  \tcc{Determine roots of original polynomials}
  Compute exponents
    $(e_{i,j})_{1\le i\le r, 0\le j < t}$ such that
    $\Gamma_i = g_0^{e_{i,0}}\cdots g_{t-1}^{e_{i,t-1}}$ for all $i$,
    using algorithm 21.2 of \citet{Ber05}
    \label{rootfind:factor}\;

  \lFor{$i\gets 1,2,\ldots, r$}{
    $\polyroots_i \gets \bigcup_{0\le j < t, e_{i,j} \ge 1} S_{0,j}$
  }
\end{algorithm}

\begin{lemma}\label{thm:rootfind}
  Given a set $\rootlist\subseteq \F$ with of size $\#\rootlist=c$
  and a list of $r$ polynomials $\Gamma_1,\ldots,\Gamma_r\in\F[z]$ each
  with degree at most $s$, \cref{alg:rootfind} \rootfind{} correctly
  determines all roots of all $\Gamma_i$'s which are in $\rootlist$ using
  $\softoh{rs+c}$ field operations
\end{lemma}
\begin{proof}
  Algorithms 18.1 and 21.2 of \citet{Ber05} are deterministic and
  compute (respectively) a coprime basis and then a factorization of the
  $\Gamma_i$'s according to the basis elements.

  The polynomials $M_{i,j}$ form a standard product tree from the
  polynomials in the coprime basis. For any element of the product tree,
  its roots must be a subset of the roots of its parent node. So the
  multi-point evaluations on \cref{rootfind:evals} determine, for every
  polynomial in the tree, which elements of $\rootlist$ are roots of that
  polynomial.

  The cost of computing the coprime basis using Algorithm 18.1 in
  \citep{Ber05} is $\softoh{rs}$ with at least a $\log^5(rs)$ factor.
  This gives the first term in the complexity (and explains our use of
  soft-oh notation throughout).
  This cost dominates the cost of constructing the product tree and
  factoring the $\Gamma_i$'s on \cref{rootfind:factor}.

  Using the standard product tree algorithm \citep{BM75},
  the cost of multi-point evaluation of a degree-$n$ polynomial
  at each of $n$ points is $\softoh{n}$ field operations.
  In our algorithm,
  this happens at each level down the product tree. At each of
  $\ceil{\log_2 t} \in O(\log(rs))$ levels there
  are at most $c$ points (since the polynomials at each level are
  relatively prime). The total degree at each level is
  at most $\deg \prod_i \Gamma_i \le rs$.
  If $rs > c$, then this cost is already bounded by that of determining
  the coprime basis. Otherwise, the multi-point evaluations occur in
  $O(c/(rs))$ batches of $rs$ points each, giving the second term in the
  complexity statement.
\end{proof}

\subsection{Batched sparse interpolation algorithm}

We are now ready to present the full algorithm for performing
simultaneous sparse interpolation on $r$ polynomials based on $2s$
evaluations each.

\begin{algorithm}[htb]
  \caption{\protect\hypertarget{spinterp}{\spinterptext{}}%
    $(r,n,s, \expons, \theta, Y)$%
    \label{alg:spinterp}}
  \KwIn{Bounds $r,n,s\in\NN$, set $\expons\subseteq\{0,\ldots,n-1\}$ of
    possible exponents, high-order element $\theta\in\F$,
    and evaluations $Y_{i,j}\in\F$
    such that $f_i(\theta^j)=Y_{i,j}$ for all $1\le i\le r$ and
    $0\le j < 2s$}
  \KwOut{Nonzero coefficients and
    corresponding exponents of $f_i\in\F[x]$ or $0$ indicating failure,
    for each $1 \le i \le r$}
  \BlankLine
  \For{$i=1,2,\ldots,r$}{
    $\Gamma_i \gets $MinPoly$(Y_{i,0},\ldots,Y_{i,2s-1})$
      \label{spinterp:minpoly}
      using fast Berlekamp-Massey algorithm
      \;
  }
  $\polyroots_1,\ldots,\polyroots_r \gets \rootfind(\{\,\theta^d \mid d\in\expons\,\},
    \Gamma_1,\ldots, \Gamma_r)$\
    \label{spinterp:rootfind}\;
  \For{$i=1,2,\ldots,r$}{
    \lIf{$\#\polyroots_i \ne \deg \Gamma_i$}{$f_i\gets 0$}
    \Else{
      $t_i \gets \#\polyroots_i$ \;
      $\{e_{i,1},\ldots,e_{i,t_i}\} \gets \{\,d\in\expons \mid \theta^d
      \in \polyroots_i\,\}$ \;
      $[a_{i,1},\ldots,a_{i,t_i}] \gets$ solution of transposed
        Vandermonde system from roots
        $[\theta^{e_{i,1}},\ldots,\theta^{e_{i,t_t}}]$
        and right-hand vector $[Y_{i,0},\ldots,Y_{i,t_i-1}]$ \;
      $f_i \gets a_{i,1}x^{e_{i,1}} + \cdots + a_{i,t_i}x^{e_{i,t_i}}$
      \;
    }
  }
\end{algorithm}

\begin{theorem}\label{thm:spinterp}
  Let $\F$ be a field with an element $\theta\in\F$ whose multiplicative
  order is at least $n$, and
  suppose $f_1,\ldots,f_r\in\F[x]$ are unknown univariate polynomials.
  Given bounds $n,s\in\NN$,
  a set $\expons\subseteq \NN$ of possible exponents,
  and evaluations $f_i(\theta^j)$
  for all $1\le i \le r$ and $0\le j < 2s$,
  \cref{alg:spinterp} \spinterp{} requires
  $\softoh{rs + \#\expons \log n}$
  field operations in the worst case.

  The algorithm correctly recovers every polynomial $f_i$ such that
  $\#\supp f_i \le s$, $\supp f_i \subseteq \expons$, and $\deg f_i <
  n$.
\end{theorem}
\begin{proof}
  From \cref{thm:rootfind}, the cost of \rootfind{} is
  $\softoh{rs+\#\expons}$. This dominates the cost of the minpoly and
  transposed Vandermonde computations, which are both $\softoh{rs}$
  \citep{KY89}. Using binary powering, computing the set
  $\{\,\theta^d\mid d\in\expons\,\}$ on \cref{spinterp:rootfind}
  requires $O(\#\expons \log n)$ field operations.

  If $f_i$ has at most $t_i\le s$ nonzero terms, then the minpoly $\Gamma_i$
  computed on \cref{spinterp:minpoly} is a degree-$t_i$ polynomial
  with exactly $t_i$ distinct roots. If $\supp(f_i) \subseteq \expons$,
  then by \cref{thm:rootfind}, \rootfind{} correctly finds all these
  roots. Because these are the only terms in $f_i$, solving the
  transposed Vandermonde system correctly determines all corresponding
  coefficients; see \citep{KY89} or \citep{HL13} for details.
\end{proof}

\section{Performing evaluations}\label{sec:evals}

Consider the matrix multiplication with errors problem, recovering an
sparse error matrix $E \in \F^{r\times n}$ according to the
equation
$E = C - AB$. In this
section, we show how to treat each row of $E$ as a sparse polynomial and
evaluate each of these at consecutive powers of a high-order element
$\theta$, as needed for the sparse interpolation algorithm from
\cref{sec:spinterp}. The same techniques will also be used in the
matrix inverse with errors algorithm.

Consider a column vector of powers of an
indeterminate, $\xx = [1,x,x^2,\ldots,x^{n-1}]^T$.
The matrix-vector product $E\xx$ then consists of $m$ polynomials, each
degree less than $n$, and with a 1-1 correspondence between nonzero
terms of the polynomials and nonzero entries in the corresponding row of
$E$.

Evaluating every polynomial in $E\xx$ at a point
$\theta\in\F$ is the same as multiplying $E$ by a vector
$\vv=[1,\theta,\theta^2,\ldots,\theta^{n-1}]^T$.
Evaluating each polynomial in $E\xx$ at the first $2s$ powers
$1,\theta,\theta^2,\ldots,\theta^{2s-1}$ allows for the unique recovery
of all $t$-sparse rows, according to \cref{thm:spinterp}.
This means multiplying $E$ times an \emph{evaluation matrix}
$V\in\F^{n\times 2s}$ such that $V_{i,j} = \theta^{ij}$.

Consider a single row vector $\uu \in \F^{1\times n}$ with $c$ nonzero
entries. Multiplying $\uu V$ means evaluating a $c$-sparse polynomial at
the first $2s$ powers of $\theta$; we can view it as removing the $n-c$
zero entries from $\uu$ to give $\uu'\in\F^{1\times c}$, and removing
the same $n-c$ rows from $V$ to get $V'\in\F^{c\times 2s}$, and then
computing the product $\uu'V'$. This evaluation matrix with removed rows
$V'$ is actually a \emph{transposed Vandermonde matrix} built from the
entries $\theta^i$ for each index $i$ of a nonzero entry in $\uu$.

An explicit algorithm for this transposed Vandermonde matrix application
is given in Section 6.2 of \citet{BLS03}, and they show that the
complexity of multiplying $\uu'V'$ is $\softoh{c + s}$ field operations,
and essentially amounts to a power series inversion and product tree traversal.
(See also Section 5.1 of \citet{HL13} for a nice description of this
approach to fast evaluation of a sparse polynomial at consecutive
powers.)
We repeat this for each row in a given matrix $M$ to compute $AV$.

Now we are ready to show how to compute $EV = (C-AB)V$ efficiently,
using the approach just discussed to compute $CV$ and $BV$ before
multiplying $A$ times $BV$ and performing a subtraction.
As we will show, the rectangular multiplication of $A$ times $BV$ is the
only step which is not softly-linear time, and this dominates the
complexity.

\begin{algorithm}[tbp]
  \caption{\protect\hypertarget{diffeval}{\diffevaltext{}}%
    $(A,B,C,\theta,s)$%
    \label{alg:diffeval}}
  \KwIn{Matrices $A\in\F^{r\times\ell}$, $B\in\F^{\ell\times n}$,
    $C\in\F^{r\times n}$, high-order element $\theta\in\F$,
    and integer $s\in\NN$}
  \KwOut{Matrix $Y=(C-AB)V \in\F^{r\times 2s}$,
    where $V$ is the $r\times 2s$
    evaluation matrix given by $(\theta^{ij})_{0\le i < n, 0\le j < 2s}$.}
  \BlankLine
  Pre-compute $\theta^j$ for all indices $j$ of nonzero columns in
  $B$ or $C$ \;
  $Y_C \gets CV$ using the pre-computed $\theta^j$'s
    and fast transposed Vandermonde applications row by row
    \label{diffeval:tvan1} \;
  $Y_B \gets BV$ similarly \label{diffeval:tvan2} \;
  $Y_{AB} \gets A Y_B$ using dense or sparse multiplication,
    whichever is faster \;
  \Return{$Y_C - Y_{AB}$}
\end{algorithm}

\begin{lemma}\label{thm:diffeval}
  \Cref{alg:diffeval} \diffeval{} always returns the correct
  matrix of evaluations $(C-AB)V$ and uses
  $$\softoh{\#B + \#C + n + s\cdot \min\left(\#A+r,
  \frac{r\ell}{\min(r,s,\ell)^{3-\omega}}\right)}$$
  field operations.
\end{lemma}
\begin{proof}
  Correctness comes from the discussion above and the correctness of the
  algorithm in section 6.2 of \citet{BLS03} which is used on
  \cref{diffeval:tvan1,diffeval:tvan2}.

  The cost of pre-computing all powers $\theta^j$ using binary powering
  is $O(n\log n)$ field operations, since the number of nonzero columns
  in $B$ or $C$ is at most $n$. Using these precomputed values,
  each transposed Vandermonde apply
  costs $\softoh{k + s}$, where $k$ is the number of
  nonzero entries in the current row. Summing over all rows of $C$ and
  $B$ gives $\softoh{\#B + \#C + rs + \ell s}$ for these steps.

  The most significant step is the rectangular multiplication of
  $A\in\F^{r\times \ell}$ times $Y_B\in\F^{\ell \times s}$.
  Using sparse multiplication, the cost is $\oh{\#A\cdot s}$, by
  \cref{thm:sparsemul}. Using dense multiplication, the cost is
  $\oh{r\ell s/\min(r,\ell,s)^{\omega-3}}$ according to
  \cref{thm:rectmul}.

  Because all the relevant parameters $\#A,s,r,\ell$
  are known before the multiplication, we assume that the underlying
  software makes the best choice among these two algorithms.

  Note that $rs$ and $\ell s$ are definitely dominated by the cost of
  the dense multiplication. In the sparse multiplication case, the
  situation is more subtle when $\#A < \ell$. But in this case, we
  just suppose that rows of $BV$ which are not used in the sparse
  multiplication $A(BV)$ are never computed, so the complexity is as
  stated.
\end{proof}

\section{Multiplication with error correction algorithm}
\label{sec:mul}

We now have all the necessary components to present the complete
multiplication algorithm, \cref{alg:mulerr} \mulerr{}.
The general idea is to repeatedly use \nzrows{} to determine the
locations on nonzeros in $E$, then \diffeval{} and \spinterp{} to
recover half of the nonzero rows at each step.

Setting the target
sparsity of each recovered row to $s=\ceil{2k/r}$, where $r$ is the
number of nonzero rows, ensures that at least half of the rows are
recovered at each step. To get around the fact that the number of errors
$k$ is initially unknown, we start with an initial guess of $k=1$, and
double this guess whenever fewer than half of the remaining rows are
recovered.

The procedure is probabilistic of the Monte Carlo type \emph{only}
because of the Monte Carlo subroutine \nzrows{} to determine which rows
are erroneous. The other parts of the algorithm --- computing
evaluations and recovering nonzero entries --- are deterministic based
on the bounds they are given.

\begin{algorithm}[tbp]
  \caption{\protect\hypertarget{mulerr}{\mulerrtext{}}%
    $(A,B,C,\theta,\epsilon)$%
    \label{alg:mulerr}}
  \KwIn{Matrices $A\in\F^{m\times\ell}$, $B\in\F^{\ell\times n}$,
    $C\in\F^{m\times n}$, high-order element $\theta\in\F$,
    and error bound $0<\epsilon<1$}
  \KwOut{Matrix $E\in\F^{m\times n}$ such that, with probability
    at least $1-\epsilon$, $AB = C-E$}
  \BlankLine
  $k \gets 1$ \;
  $E \gets \mathbf{0}^{m\times n}$ \;
  $\epsilon' \gets \epsilon / (4\ceil{\log_2 (mn)} + 1)$ \;
  $\rowinds \gets \nzrows(V \mapsto CV - A(BV), \epsilon')$ \;
  \While{$\#\rowinds \ge 1$}{
    \If{$\#\nzrows(V\mapsto (V^T(C-E) - (V^T A)B)^T, \epsilon')
        > \#\rowinds$}{
      Transpose $C$ and $E$, swap and transpose $A$ and $B$,
      replace $\rowinds$ \;
    }
    $r \gets \#\rowinds$ \;
    $A' \gets$ submatrix of $A$ from rows in $\rowinds$ \;
    $C' \gets$ submatrix of $(C-E)$ from rows in $\rowinds$ \;
    $s \gets \ceil{2(k-\#E)/r}$ \;
    $Y \gets \diffeval(A',B,C',\theta,s)$ \;
    $f_1,\ldots,f_{r} \gets \spinterp(r,n,s,\theta,Y)$ \;
    \For{$i \gets 1,2,\ldots,r$}{
      Set $(\rowinds_i,e)$th entry of $E$ to $c$ for each
      term $cx^e$ of $f_i$ \;
    }
    $\rowinds \gets \nzrows(V\mapsto (C-E)V - A(BV), \epsilon')$ \;
    \If{$\#\rowinds > r/2$}{
      $k \gets 2k$ \;
      \lIf{$k \ge 2n\#J$}{\Return{$C-AB$}\label{mulerr:giveup}}
    }
    \ForEach{$i \in \rowinds$}{
      Clear entries from row $i$ of $E$ added on this iteration \;
    }
  }
  \Return{$E$}
\end{algorithm}

\begin{theorem}\label{thm:mulerr}
  With probability at least $1-\epsilon$,
  \cref{alg:mulerr} \mulerr{} finds all errors in $C$
  and uses
  $$\softoh{m+\ceil{\log_{\#\F}\tfrac{1}{\epsilon}}(n+t)
    + k\cdot \min\left(
      \ceil{\frac{t}{r}},\,
        \ell/\min(\ell,r,\tfrac{k}{r})^{3-\omega}\right)}$$
  field operations,
  where $k$ is the \emph{actual} number of errors in the given product
  $C$. Otherwise, it uses
  $\softoh{m\ell n/\min(m,\ell,n)^{3-\omega}}$ field operations and may
  return an incorrect result.
\end{theorem}
\begin{proof}
  The only probabilistic
  parts of the algorithm are the calls to \nzrows{}, which can only fail
  by incorrectly returning too few rows. If this never happens, then
  each iteration through the while loop either (a) discovers that the
  value of $k$ was too small and increases it, or (b) recovers half of
  the rows or columns of $E$.

  So the total number of iterations of the
  while loop if the calls to \nzrows{} are never incorrect is at most
  $\ceil{\log_2 k} + \ceil{\log_2 m} + \ceil{\log_2 n}.$
  By the union bound, from the way that $\epsilon'$ is computed and the
  fact that there are two calls to \nzrows{} on every iteration, the
  probability that \nzrows{} is never incorrect is at least
  $1-\epsilon$.

  In this case, the rest of the correctness and the running time follow
  directly from \cref{thm:nzrows,thm:diffeval,thm:spinterp}.

  If a call to \nzrows{} returns an incorrect result, two bad things can
  happen. First, if $\rowinds$ is too small during some iteration, all
  of the evaluations may be incorrect, leading to $k$ being increased.
  In extremely unlucky circumstances, this may happen $O(\log (rn))$
  times until the product is na\"ively recomputed on
  \cref{mulerr:giveup}, leading to the worst-case running time.

  The second bad thing that can occur when a call to \nzrows{} fails is
  that it may incorrectly report all errors have been found, leading to
  the small $1-\epsilon$ chance that the algorithm returns an incorrect
  result.
\end{proof}

\section{Inverse error correction algorithm}\label{sec:inverse}

As discussed in \cref{sec:over:inverse}, our algorithm for correcting
errors in a matrix inverse follows mostly the same outline as that for
correcting a matrix product, with two important changes. These are the
basis for the next two lemmas.

\begin{lemma}\label{thm:invnz}
  For any matrices $E,A\in\F^{n\times n}$ where $A$ is invertible,
  a call to $\nzrows(V\mapsto EAV,\epsilon)$
  correctly returns the nonzero rows of $E$ with probability
  at least $1-\epsilon$.
\end{lemma}
\begin{proof}
  Recall from the proof of \cref{thm:nzrows} that the correctness of
  \nzrows{} depends on applying a matrix $V\in\F^{n\times\ell}$ of
  random entries to the unknown matrix $E$.

  In the current formulation, instead
  of applying the random matrix $V$ directly to $E$, instead the product
  $AV$ is applied. But because $A$ is nonsingular, there is a 1-1
  correspondence between the set of all matrices $V\in\F^{n\times\ell}$ and
  matrices the set $\{\,AV \mid V\in\F^{n\times\ell}\,\}$.
  That is, applying a random matrix to $EA$ is the same as applying a
  random matrix to $E$ itself, and therefore the same results hold.
\end{proof}

\begin{lemma}\label{thm:lowrank}
Given any rank-$r$ matrix $A\in\F^{n\times r}$,
it is possible to compute a matrix $X\in\F^{r\times r}$ formed from a
subset of the rows of $A$, and its inverse $X^{-1}$, using
$\softoh{\#A + r^\omega}$ field operations.
\end{lemma}
\begin{proof}
  This is a direct consequence of \citep[Theorem 2.11]{CKL13},
  along with the classic algorithm of \citep{BH74} for fast matrix
  inversion. Alternatively, one could use the approach of
  \citep{SY15} to compute the lexicographically minimal set of linearly
  independent rows in $M$, as well as a representation of the inverse, in
  the same running time.
\end{proof}

The resulting algorithm for matrix inversion with errors is presented in
\cref{alg:inverserr} \inverserr{}.

\begin{algorithm}[tbp]
  \caption{\protect\hypertarget{inverserr}{\inverserrtext{}}%
    $(A,B,\theta,\epsilon)$%
    \label{alg:inverserr}}
  \KwIn{Matrices $A,B\in\F^{n\times n}$,
    high-order element $\theta\in\F$,
    and error bound $0<\epsilon<1$}
  \KwOut{Matrix $E\in\F^{n\times n}$ such that, with probability
    at least $1-\epsilon$, $A^{-1} = B+E$}
  \BlankLine
  $k \gets 1$ \;
  $E \gets \mathbf{0}^{n\times n}$ \;
  $\epsilon' \gets \epsilon / (8\ceil{\log_2 n} + 1)$ \;
  $\rowinds \gets \nzrows(V \mapsto V - B(AV), \epsilon')$
    \label{inverserr:nzr1}\;
  \While{$\#\rowinds \ge 1$}{
    \If{$\#\nzrows(V\mapsto V - A((B+E)V), \epsilon')
        > \#\rowinds$}{
      Transpose $A$, $B$, and $E$, and replace $\rowinds$ \;
    }
    $r \gets \#\rowinds$ \;
    $A' \gets$ submatrix of $A$ from columns in $\rowinds$ \;
    $\rowinds' \gets$ set of $r$ linearly independent rows in $A'$ according to
      \cref{thm:lowrank} \;
    $I', A'' \gets$ submatrices of $I,A$ from the rows chosen in
      $\rowinds'$ \;
    $X^{-1} \gets$ inverse of submatrix of $A'$ according to
      \cref{thm:lowrank} \;
    $s \gets \ceil{2(k-\#E)/r}$ \;
    $Y \gets \diffeval(A'',B,I',\theta,s)$ \label{inverserr:diffeval}\;
    $Y' \gets X^{-1} Y$ \;
    $f_1,\ldots,f_{r} \gets \spinterp(r,n,s,\theta,Y')$ \;
    \For{$i \gets 1,2,\ldots,r$}{
      Set $(\rowinds_i,e)$th entry of $E$ to $c$ for each
      term $cx^e$ of $f_i$ \;
    }
    $\rowinds \gets \nzrows(V \mapsto V - (B+E)(AV), \epsilon')$
      \label{inverserr:nzr2}\;
    \If{$\#\rowinds > r/2$}{
      $k \gets 2k$ \;
      \lIf{$k \ge 2n\#J$}{\Return{$A^{-1}-B$}\label{inverserr:giveup}}
    }
    \ForEach{$i \in \rowinds$}{
      Clear entries from row $i$ of $E$ added on this iteration \;
    }
  }
  \Return{$E$}
\end{algorithm}

\begin{theorem}\label{thm:inverserr}
  With probability at least $1-\epsilon$,
  \cref{alg:inverserr} \inverserr{} finds all errors in $B$
  and uses
  $$\softoh{\ceil{\log_{\#\F}\tfrac{1}{\epsilon}}t
    + nk/\min(r,\tfrac{k}{r})^{3-\omega}
    + r^\omega}$$
  field operations,
  where $k$ is the \emph{actual} number of errors in the given product
  $C$. Otherwise, it uses $\softoh{n^\omega}$ field operations and may
  return an incorrect result.
\end{theorem}
\begin{proof}
  In this algorithm, we make use of two different formulas for $E$:
  \begin{align}
    EA &= I - BA \label{eqn:ea} \\
    AE &= I - AB \label{eqn:ae}
  \end{align}

  The first formula \eqref{eqn:ea} is used to determine the nonzero rows
  of $E$ on \cref{inverserr:nzr1,inverserr:nzr2}, and the second
  formula \eqref{eqn:ae} is used to evaluate the rows of $XE$ on
  \cref{inverserr:diffeval}.

  The main difference in this algorithm compared to \mulerr{} is the
  need to find $r$ nonzero rows of $A'$ to constitute the matrix $X$ and
  compute its inverse $X^{-1}$. According to \cref{thm:lowrank} these
  both cost $\softoh{r^\omega}$, which gives the additional term in the
  complexity statement. Note that this also eliminates the utility of
  sparse multiplication in all cases, simplifying the complexity
  statement somewhat compared to that of \mulerr{}.

  The rest of the proof is identical to that of \cref{thm:mulerr}.
\end{proof}

\iftoggle{sage}{
\section{Sage implementation}\label{sec:imp}
}{}

\newcommand{\Gathen}{\relax}\newcommand{\Hoeven}{\relax}


\newcommand{\Gathen}{\relax}\newcommand{\Hoeven}{\relax}
\begin{thebibliography}{36}
\providecommand{\natexlab}[1]{#1}
\providecommand{\url}[1]{\texttt{#1}}
\expandafter\ifx\csname urlstyle\endcsname\relax
  \providecommand{\doi}[1]{doi: #1}\else
  \providecommand{\doi}{doi: \begingroup \urlstyle{rm}\Url}\fi

\bibitem[Amossen and Pagh(2009)]{AP09}
Rasmus~Resen Amossen and Rasmus Pagh.
\newblock Faster join-projects and sparse matrix multiplications.
\newblock In \emph{Proceedings of the 12th International Conference on Database
  Theory}, ICDT '09, pages 121--126, New York, NY, USA, 2009. ACM.
\newblock \doi{10.1145/1514894.1514909}.

\bibitem[Arnold et~al.(2015)Arnold, Giesbrecht, and Roche]{AGR15}
Andrew Arnold, Mark Giesbrecht, and Daniel~S. Roche.
\newblock Faster sparse multivariate polynomial interpolation of straight-line
  programs.
\newblock \emph{Journal of Symbolic Computation}, 2015.
\newblock \doi{10.1016/j.jsc.2015.11.005}.

\bibitem[Ben-Or and Tiwari(1988)]{BT88}
Michael Ben-Or and Prasoon Tiwari.
\newblock A deterministic algorithm for sparse multivariate polynomial
  interpolation.
\newblock In \emph{Proceedings of the twentieth annual ACM symposium on Theory
  of computing}, STOC '88, pages 301--309, New York, NY, USA, 1988. ACM.
\newblock \doi{10.1145/62212.62241}.

\bibitem[Bernstein(2005)]{Ber05}
Daniel~J. Bernstein.
\newblock Factoring into coprimes in essentially linear time.
\newblock \emph{Journal of Algorithms}, 54\penalty0 (1):\penalty0 1--30, 2005.
\newblock \doi{10.1016/j.jalgor.2004.04.009}.

\bibitem[B\"ohm et~al.(2015)B\"ohm, Decker, Fieker, and Pfister]{BDFP15}
Janko B\"ohm, Wolfram Decker, Claus Fieker, and Gerhard Pfister.
\newblock The use of bad primes in rational reconstruction.
\newblock \emph{Math. Comp.}, 84\penalty0 (296):\penalty0 3013--3027, 2015.
\newblock \doi{10.1090/mcom/2951}.

\bibitem[Borodin and Munro(1975)]{BM75}
A.~Borodin and I.~Munro.
\newblock \emph{The computational complexity of algebraic and numeric
  problems}.
\newblock Number~1 in Elsevier Computer Science Library; Theory of Computation
  Series. American Elsevier Pub. Co., New York, 1975.

\bibitem[Bostan et~al.(2003)Bostan, Lecerf, and Schost]{BLS03}
A.~Bostan, G.~Lecerf, and \'{E}. Schost.
\newblock Tellegen's principle into practice.
\newblock In \emph{Proceedings of the 2003 International Symposium on Symbolic
  and Algebraic Computation}, ISSAC '03, pages 37--44. ACM, 2003.
\newblock \doi{10.1145/860854.860870}.

\bibitem[Boyer and Kaltofen(2014)]{BK14}
Brice Boyer and Erich~L. Kaltofen.
\newblock Numerical linear system solving with parametric entries by error
  correction.
\newblock In \emph{Proceedings of the 2014 Symposium on Symbolic-Numeric
  Computation}, SNC '14, pages 33--38. ACM, 2014.
\newblock \doi{10.1145/2631948.2631956}.

\bibitem[Bunch and Hopcroft(1974)]{BH74}
James~R. Bunch and John~E. Hopcroft.
\newblock Triangular factorization and inversion by fast matrix multiplication.
\newblock \emph{Mathematics of Computation}, 28\penalty0 (125):\penalty0
  231--236, 1974.
\newblock URL \url{http://www.jstor.org/stable/2005828}.

\bibitem[Cantor and Kaltofen(1991)]{CK91}
David~G. Cantor and Erich Kaltofen.
\newblock On fast multiplication of polynomials over arbitrary algebras.
\newblock \emph{Acta Informatica}, 28:\penalty0 693--701, 1991.
\newblock \doi{10.1007/BF01178683}.

\bibitem[Cheung et~al.(2013)Cheung, Kwok, and Lau]{CKL13}
Ho~Yee Cheung, Tsz~Chiu Kwok, and Lap~Chi Lau.
\newblock Fast matrix rank algorithms and applications.
\newblock \emph{J. ACM}, 60\penalty0 (5):\penalty0 31:1--31:25, October 2013.
\newblock \doi{10.1145/2528404}.

\bibitem[Comer et~al.(2012)Comer, Kaltofen, and Pernet]{CKP12}
Matthew~T. Comer, Erich~L. Kaltofen, and Cl{\'e}ment Pernet.
\newblock Sparse polynomial interpolation and {Berlekamp/Massey} algorithms
  that correct outlier errors in input values.
\newblock In \emph{Proceedings of the 37th International Symposium on Symbolic
  and Algebraic Computation}, ISSAC '12, pages 138--145, New York, NY, USA,
  2012. ACM.
\newblock \doi{10.1145/2442829.2442852}.

\bibitem[Dumas and Kaltofen(2014)]{DK14}
Jean-Guillaume Dumas and Erich Kaltofen.
\newblock Essentially optimal interactive certificates in linear algebra.
\newblock In \emph{Proceedings of the 39th International Symposium on Symbolic
  and Algebraic Computation}, ISSAC '14, pages 146--153. ACM, 2014.
\newblock \doi{10.1145/2608628.2608644}.

\bibitem[Dumas et~al.(2008)Dumas, Giorgi, and Pernet]{DGP08}
Jean-Guillaume Dumas, Pascal Giorgi, and Cl\'{e}ment Pernet.
\newblock Dense linear algebra over word-size prime fields: the {FFLAS} and
  {FFPACK} packages.
\newblock \emph{ACM Trans. Math. Softw.}, 35:\penalty0 19:1--19:42, October
  2008.
\newblock \doi{10.1145/1391989.1391992}.

\bibitem[Dumas et~al.(2016)Dumas, Kaltofen, Thom{\'e}, and Villard]{DKTV16}
Jean-Guillaume Dumas, Erich Kaltofen, Emmanuel Thom{\'e}, and Gilles Villard.
\newblock Linear time interactive certificates for the minimal polynomial and
  the determinant of a sparse matrix.
\newblock In \emph{Proceedings of the ACM on International Symposium on
  Symbolic and Algebraic Computation}, ISSAC '16, pages 199--206. ACM, 2016.
\newblock \doi{10.1145/2930889.2930908}.

\bibitem[Dumas et~al.(2017)Dumas, Lucas, and Pernet]{DLP17}
Jean-Guillaume Dumas, David Lucas, and Cl{\'e}ment Pernet.
\newblock Certificates for triangular equivalence and rank profiles.
\newblock In \emph{Proceedings of the 2017 ACM on International Symposium on
  Symbolic and Algebraic Computation}, ISSAC '17, pages 133--140. ACM, 2017.
\newblock \doi{10.1145/3087604.3087609}.

\bibitem[Freivalds(1979)]{Fre79}
R{\={u}}si{\c{n}}{\v{s}} Freivalds.
\newblock Fast probabilistic algorithms.
\newblock In Ji{\v{r}}{\'i} Be{\v{c}}v{\'a}{\v{r}}, editor, \emph{Mathematical
  Foundations of Computer Science 1979}, pages 57--69. Springer Berlin
  Heidelberg, 1979.

\bibitem[Gasieniec et~al.(2017)Gasieniec, Levcopoulos, Lingas, Pagh, and
  Tokuyama]{GLL+17}
Leszek Gasieniec, Christos Levcopoulos, Andrzej Lingas, Rasmus Pagh, and
  Takeshi Tokuyama.
\newblock Efficiently correcting matrix products.
\newblock \emph{Algorithmica}, 79\penalty0 (2):\penalty0 428--443, Oct 2017.
\newblock \doi{10.1007/s00453-016-0202-3}.

\bibitem[Goldreich et~al.(1999)Goldreich, Ron, and Sudan]{GRS99}
Oded Goldreich, Dana Ron, and Madhu Sudan.
\newblock Chinese remaindering with errors.
\newblock In \emph{Proceedings of the Thirty-first Annual ACM Symposium on
  Theory of Computing}, STOC '99, pages 225--234. ACM, 1999.
\newblock \doi{10.1145/301250.301309}.

\bibitem[Harvey et~al.(2017)Harvey, \Hoeven{van der Hoeven}, and Lecerf]{HHL17}
David Harvey, Joris \Hoeven{van der Hoeven}, and Gr{\'e}goire Lecerf.
\newblock Faster polynomial multiplication over finite fields.
\newblock \emph{J. ACM}, 63\penalty0 (6):\penalty0 52:1--52:23, January 2017.
\newblock \doi{10.1145/3005344}.

\bibitem[\Hoeven{van der Hoeven} and Lecerf(2013)]{HL13}
Joris \Hoeven{van der Hoeven} and Gr\'egoire Lecerf.
\newblock On the bit-complexity of sparse polynomial and series multiplication.
\newblock \emph{Journal of Symbolic Computation}, 50:\penalty0 227--0254, 2013.
\newblock \doi{10.1016/j.jsc.2012.06.004}.

\bibitem[Huang and Gao(2017)]{HG17}
Qiao{-}Long Huang and Xiao{-}Shan Gao.
\newblock Faster deterministic sparse interpolation algorithms for
  straight-line program multivariate polynomials.
\newblock \emph{CoRR}, abs/1709.08979, 2017.
\newblock URL \url{http://arxiv.org/abs/1709.08979}.

\bibitem[Javadi and Monagan(2010)]{JM10}
Seyed Mohammad~Mahdi Javadi and Michael Monagan.
\newblock Parallel sparse polynomial interpolation over finite fields.
\newblock In \emph{Proceedings of the 4th International Workshop on Parallel
  and Symbolic Computation}, PASCO '10, pages 160--168, New York, NY, USA,
  2010. ACM.
\newblock \doi{10.1145/1837210.1837233}.

\bibitem[Kaltofen and Yagati(1989)]{KY89}
Erich Kaltofen and Lakshman Yagati.
\newblock Improved sparse multivariate polynomial interpolation algorithms.
\newblock In P.~Gianni, editor, \emph{Symbolic and Algebraic Computation},
  volume 358 of \emph{Lecture Notes in Computer Science}, pages 467--474.
  Springer Berlin / Heidelberg, 1989.
\newblock \doi{10.1007/3-540-51084-2_44}.

\bibitem[Kaltofen(2010)]{Kal10a}
Erich~L. Kaltofen.
\newblock Fifteen years after {DSC} and {WLSS2}: {W}hat parallel computations
  {I} do today [invited lecture at {PASCO} 2010].
\newblock In \emph{Proceedings of the 4th International Workshop on Parallel
  and Symbolic Computation}, PASCO '10, pages 10--17, New York, NY, USA, 2010.
  ACM.
\newblock \doi{10.1145/1837210.1837213}.

\bibitem[Kaltofen and Yang(2013)]{KY13}
Erich~L. Kaltofen and Zhengfeng Yang.
\newblock Sparse multivariate function recovery from values with noise and
  outlier errors.
\newblock In \emph{Proceedings of the 38th International Symposium on Symbolic
  and Algebraic Computation}, ISSAC '13, pages 219--226. ACM, 2013.
\newblock \doi{10.1145/2465506.2465524}.

\bibitem[Kaltofen et~al.(2011)Kaltofen, Nehring, and Saunders]{KNS11}
Erich~L. Kaltofen, Michael Nehring, and B.~David Saunders.
\newblock Quadratic-time certificates in linear algebra.
\newblock In \emph{Proceedings of the 36th International Symposium on Symbolic
  and Algebraic Computation}, ISSAC '11, pages 171--176. ACM, 2011.
\newblock \doi{10.1145/1993886.1993915}.

\bibitem[Kaltofen et~al.(2017)Kaltofen, Pernet, Storjohann, and
  Waddell]{KPSW17}
Erich~L. Kaltofen, Cl{\'e}ment Pernet, Arne Storjohann, and Cleveland Waddell.
\newblock Early termination in parametric linear system solving and rational
  function vector recovery with error correction.
\newblock In \emph{Proceedings of the 2017 ACM on International Symposium on
  Symbolic and Algebraic Computation}, ISSAC '17, pages 237--244. ACM, 2017.
\newblock \doi{10.1145/3087604.3087645}.

\bibitem[Khonji et~al.(2010)Khonji, Pernet, Roch, Roche, and Stalinski]{KPR+10}
Majid Khonji, Cl{\'e}ment Pernet, Jean-Louis Roch, Thomas Roche, and Thomas
  Stalinski.
\newblock Output-sensitive decoding for redundant residue systems.
\newblock In \emph{Proceedings of the 2010 International Symposium on Symbolic
  and Algebraic Computation}, ISSAC '10, pages 265--272. ACM, 2010.
\newblock \doi{10.1145/1837934.1837985}.

\bibitem[Le~Gall(2012)]{LGal12}
Fran\c{c}ois Le~Gall.
\newblock Faster algorithms for rectangular matrix multiplication.
\newblock In \emph{2012 IEEE 53rd Annual Symposium on Foundations of Computer
  Science}, pages 514--523, Oct 2012.
\newblock \doi{10.1109/FOCS.2012.80}.

\bibitem[Le~Gall(2014)]{LGal14}
Fran\c{c}ois Le~Gall.
\newblock Powers of tensors and fast matrix multiplication.
\newblock In \emph{Proceedings of the 39th International Symposium on Symbolic
  and Algebraic Computation}, ISSAC '14, pages 296--303, New York, NY, USA,
  2014. ACM.
\newblock \doi{10.1145/2608628.2608664}.

\bibitem[Lingas(2009)]{Lin09}
Andrzej Lingas.
\newblock A fast output-sensitive algorithm for boolean matrix multiplication.
\newblock In Amos Fiat and Peter Sanders, editors, \emph{Algorithms - ESA
  2009}, pages 408--419. Springer Berlin Heidelberg, 2009.
\newblock \doi{10.1007/978-3-642-04128-0_37}.

\bibitem[Pagh(2013)]{Pag13}
Rasmus Pagh.
\newblock Compressed matrix multiplication.
\newblock \emph{ACM Trans. Comput. Theory}, 5\penalty0 (3):\penalty0 9:1--9:17,
  August 2013.
\newblock \doi{10.1145/2493252.2493254}.

\bibitem[Storjohann and Yang(2015)]{SY15}
Arne Storjohann and Shiyun Yang.
\newblock A relaxed algorithm for online matrix inversion.
\newblock In \emph{Proceedings of the 2015 ACM on International Symposium on
  Symbolic and Algebraic Computation}, ISSAC '15, pages 339--346, New York, NY,
  USA, 2015. ACM.
\newblock \doi{10.1145/2755996.2756672}.

\bibitem[Wang et~al.(2013)Wang, Zhang, Zhang, and Yi]{WZZY13}
Qian Wang, Xianyi Zhang, Yunquan Zhang, and Qing Yi.
\newblock Augem: Automatically generate high performance dense linear algebra
  kernels on x86 cpus.
\newblock In \emph{Proceedings of the International Conference on High
  Performance Computing, Networking, Storage and Analysis}, SC '13, pages
  25:1--25:12. ACM, 2013.
\newblock \doi{10.1145/2503210.2503219}.

\bibitem[Yuster and Zwick(2005)]{YZ05}
Raphael Yuster and Uri Zwick.
\newblock Fast sparse matrix multiplication.
\newblock \emph{ACM Trans. Algorithms}, 1\penalty0 (1):\penalty0 2--13, July
  2005.
\newblock \doi{10.1145/1077464.1077466}.

\end{thebibliography}
\end{document}